\documentclass[11pt]{article}

\usepackage{amsmath,amssymb,amsthm,mathtools}
\usepackage[margin=1in]{geometry}
\usepackage{enumitem}
\usepackage[colorlinks=true,linkcolor=blue,citecolor=blue,urlcolor=blue]{hyperref}

\newtheorem{theorem}{Theorem}
\newtheorem{proposition}[theorem]{Proposition}
\newtheorem{lemma}[theorem]{Lemma}
\newtheorem{corollary}[theorem]{Corollary}
\newtheorem{remark}[theorem]{Remark}

\newcommand{\F}{\mathbb F}
\newcommand{\Prb}{\Pr}
\newcommand{\E}{\mathbb E}
\newcommand{\M}{\mathsf M}
\newcommand{\Span}{\operatorname{span}}
\newcommand{\rank}{\operatorname{rank}}
\newcommand{\nul}{\operatorname{nul}}
\newcommand{\Load}{\operatorname{load}}

\title{A Note on Second-Order Expected Maximum-Load Bounds for\\ Binary Linear Hashing}
\author{Nader H. Bshouty\\ Technion}
\date{}

\begin{document}
\maketitle

\begin{abstract}
Let \(S\subseteq \F_2^u\) have size \(n=2^\ell\), and let
\(h:\F_2^u\to \F_2^\ell\) be a uniformly random linear map. For
\(y\in\F_2^\ell\), write \(\Load_h(y):=|h^{-1}(y)\cap S|\), and let
\(\M(S,h):=\max_{y\in\F_2^\ell}\Load_h(y)\) be the maximum load. Jaber, Kumar
and Zuckerman (STOC 2025) proved that the expected maximum load of \(h\) on
\(S\) is at most \(16\log n/\log\log n\), matching the fully independent
keys-into-bins scale up to constants. Their proof also gives the tail estimate
\[
 \Prb\left[
 \M(S,h)\ge R\frac{\log n}{\log\log n}
 \right]
 \le O\left(\frac{1}{R^{2}}\right).
\]
We record a base optimization in their exponential-potential method showing
that binary linear hashing nearly matches fully independent hashing also at the
level of the second-order maximum-load scale. For every \(R>1\) satisfying
\(R\ell^{1-1/R}\ge D\ln\ell\), where \(D\) is an absolute constant, we prove
\[
 \Prb\left[
 \M(S,h)\ge R\frac{\log n}{\log\log n}
 \right]
 \le
 O\left(
 \frac{(\log\log n)^2}{R^2(\log n)^{2-2/R}}
 \right).
\]
Integrating this tail yields
\[
 \E[\M(S,h)]
 \le
 \left(
 1+
 (1+o(1))
 \frac{\log\log\log n}{\log\log n}
 \right)
 \frac{\log n}{\log\log n}.
\]
Thus binary linear hashing matches fully independent hashing in the leading
term and matches the dominant second-order correction up to a \(1+o(1)\)
factor.

We also prove, by an independent self-contained argument, a sharp tail bound for
one prescribed bucket: for fixed \(y\in\F_2^\ell\),
\[
 \Prb[\Load_h(y)>2^a-2]\le \gamma^{-1}2^{-a^2},
\]
where \(\gamma=\prod_{j\ge1}(1-2^{-j})\). A subspace construction shows that this
is asymptotically tight even in the leading constant as \(a\to\infty\).
However, this controls only a fixed bucket; a direct union bound over all
buckets loses a factor \(2^\ell\).
\end{abstract}
\newpage
\tableofcontents
\section{Introduction}

Hashing is often used to distribute a set of keys among bins. If the hash
function maps many keys to the same bin, then operations such as lookup with
chaining may take a long time. Thus a central quantity is the maximum load, the
largest number of keys mapped to any single bin.

In this note we study this question for binary linear hashing. Let \(u\ge\ell\),
let \(h:\F_2^u\to\F_2^\ell\) be a uniformly random linear map, and let
\(S\subseteq\F_2^u\) have size \(n=2^\ell\). For \(y\in\F_2^\ell\), define
\[
 \Load_h(y):=|h^{-1}(y)\cap S|,
\]
and define the maximum load
\[
 \M(S,h):=\max_{y\in\F_2^\ell}\Load_h(y).
\]
Throughout this paper, all logarithms are base \(2\). We call the elements of \(S\) {\it keys} and the elements of \(\F_2^\ell\) {\it bins}.
For \(y\in\F_2^\ell\), the {\it bucket} at \(y\) is the set
 $h^{-1}(y)\cap S$.

For fully independent hashing of \(n\) keys into \(n\) bins, the expected
maximum load is
\[
 (1+o(1))\frac{\log n}{\log\log n}.
\]
More precisely, the classical second-order asymptotic is
\[
 \E[M_n]
 =
 \frac{\log n}{\log\log n}
 +
 (1+o(1))
 \frac{\log n\cdot\log\log\log n}{(\log\log n)^2}.
\]
A uniformly random linear map is much simpler than a fully random function, but
it is only pairwise independent. Therefore the fully independent keys-into-bins
analysis does not directly apply.

\subsection{The expected max-load of a linear map}

The maximum-load behavior of linear hashing has been studied for several
decades. Alon, Dietzfelbinger, Miltersen, Petrank and Tardos~\cite{ADMP97} proved that, for
binary linear hashing,
\[
 \E_h[\M(S,h)]=O(\log n\log\log n).
\]
They asked whether the fully independent scale \(O(\log n/\log\log n)\) is also
valid for random binary linear maps.

Jaber, Kumar and Zuckerman \cite{JKZ25} then proved the optimal expected maximum-load bound
for the binary case. They showed that
\[
 \E_h[\M(S,h)]\le 16\frac{\log n}{\log\log n}.
\]
They also proved the tail estimate
\[
 \Prb\left[
 \M(S,h)\ge R\frac{\log n}{\log\log n}
 \right]
 \le
 O\left(\frac1{R^2}\right).
\]
Their proof uses an exponential-potential method. One advantage of this method
is that it directly tracks the growth of bucket loads as the kernel of the
linear map is revealed one dimension at a time.

\subsection{Our result and comparison with fully independent hashing}

We record a small optimization of the Jaber--Kumar--Zuckerman potential
argument. Their method uses an exponential potential with a certain base. To
detect a bucket of load
\[
 R\frac{\log n}{\log\log n},
\]
it is enough to use base roughly \(\ell^{1/R}\), rather than base roughly
\(\ell\). This lowers the initial potential and gives a sharper tail bound.

Our main tail estimate is the following. There are absolute constants
\(C,D>0\) such that, for every \(R>1\) satisfying
\[
 R\ell^{1-1/R}\ge D\ln\ell,
\]
we have
\[
 \Prb\left[
 \M(S,h)\ge R\frac{\log n}{\log\log n}
 \right]
 \le
 C
 \frac{(\log\log n)^2}
 {R^2(\log n)^{2-2/R}}.
\]
This improves the \(O(R^{-2})\) tail bound in the range where \(R\) is fixed,
and it remains useful even when \(R\) is very close to \(1\).

Integrating the optimized tail gives
\[
 \E_h[\M(S,h)]
 \le
 \frac{\log n}{\log\log n}
 +
 (1+o(1))
 \frac{\log n\cdot\log\log\log n}{(\log\log n)^2}.
\]
Thus binary linear hashing matches fully independent hashing in the leading
term and also matches the dominant part of the second-order correction. 

\subsection{A sharp tail bound for one prescribed bucket}

We also prove a separate fixed-bucket result. This part is independent of the
potential method.

Fix one bucket \(y\in\F_2^\ell\). We prove that
\[
 \Prb[\Load_h(y)>2^a]
 \le
 \gamma^{-1}2^{-a^2},
 \qquad
 \gamma:=\prod_{j=1}^{\infty}(1-2^{-j}).
\]
The proof is elementary. If a prescribed bucket is large, then it contains many
ordered linearly independent tuples. On the other hand, a random linear map
sends any fixed independent tuple to that prescribed bucket with very small
probability. Counting such tuples and applying Markov's inequality gives the
tail bound.

This fixed-bucket bound is essentially sharp. A subspace example shows that
the upper bound is asymptotically tight even in the leading constant as
\(a\to\infty\). Thus the \(2^{-a^2}\) fixed-bucket tail is the correct scale.

However, this fixed-bucket theorem does not by itself imply the correct
expectation bound for the maximum load. A direct union bound over all
\(2^\ell\) buckets loses a factor \(2^\ell\). Therefore the fixed-bucket
estimate is best viewed as a sharp far-tail result for one prescribed bucket,
while the optimized potential method is needed to control the maximum load near
the expectation scale.

\subsection{The case of \(m\) keys and \(n\) bins}

The same base-optimized argument also applies when the number of keys is not
equal to the number of bins. Let the number of bins be \(n=2^\ell\), let
\(|S|=m\), and write
\[
 \lambda:=\frac mn
\]
for the average load. We consider the regime where the maximum-load scale is
much larger than the average load. For fully independent hashing, this scale is
the value \(t=t(m,n)\) defined by
\[
 t\ln\left(\frac{t}{e\lambda}\right)=\ln n.
\]
Thus the regime we consider is
 ${t}/{\lambda}\to\infty$,
meaning that the largest bucket is expected to be much larger than the average
bucket.

Our method gives the following analogue of the \(m=n\) result. There are
absolute constants \(C,D>0\) such that, for every threshold \(T\) satisfying
 $T\ge D\lambda n^{1/T}$,
one has
\[
 \Pr_h[\M(S,h)\ge T]
 \le
 C\left(\frac{\lambda n^{1/T}}{T}\right)^2.
\]
Taking \(T=Rt\), this becomes
\[
 \Pr_h[\M(S,h)\ge Rt]
 \le
 \frac{C}{R^2}
 \left(\frac{\lambda}{t}\right)^{2-2/R},
\]
for every \(R>1\) satisfying
\[
 R\left(\frac{t}{\lambda}\right)^{1-1/R}\ge D.
\]
Thus, when \(t/\lambda\to\infty\), the tail bound is stronger than the
constant-factor \(O(R^{-2})\) tail: for fixed \(R>1\), the extra factor
\[
 \left(\frac{\lambda}{t}\right)^{2-2/R}
\]
tends to zero.

Integrating the tail gives
\[
 \E_h[\M(S,h)]\le (1+o(1))t
\]
whenever \(t/\lambda\to\infty\). Thus, in the sparse large-load regime, binary
linear hashing matches the fully independent expected maximum-load scale up to
a \(1+o(1)\) factor. 

\subsection{Fixed buckets when the number of keys is arbitrary}

The fixed-bucket estimate also extends to the case of \(m\) keys and \(n\) bins.
Let \(n=2^\ell\), let \(S\subseteq\F_2^d\setminus\{0\}\) be a set of \(m\)
distinct nonzero vectors, and let \(h:\F_2^d\to\F_2^\ell\) be a uniformly
random linear map. For a fixed bucket \(y\in\F_2^\ell\), write
 $Z_y=|\{x\in S:h(x)=y\}|$,
 and
 $\lambda:= m/n$.
Then, for every integer \(r\ge0\), with
 $a=\left\lceil\log(r+2)\right\rceil$,
we prove
\[
 \Prb[Z_y>r]
 \le
 \lambda^a
 \left(\prod_{j=0}^{a-1}(r+2-2^j)\right)^{-1}.
\]
In particular, for thresholds of the form \(2^a-2\),
\[
 \Prb[Z_y>2^a-2]
 \le
 \gamma^{-1}\lambda^a2^{-a^2},
 \qquad
 \gamma:=\prod_{j=1}^{\infty}(1-2^{-j}).
\]

This recovers the balanced fixed-bucket estimate when \(m=n\), since then
\(\lambda=1\). For general \(m\), the extra factor \(\lambda^a\) reflects the
average load of the prescribed bucket. We also show that this dependence is
sharp up to absolute constants: when \(m=2^d\), a subspace construction gives
a matching lower bound of order
 $\lambda^a2^{-a^2}$
for the zero bucket, in the natural range \(a\ge d-\ell\). Thus the
fixed-bucket tail has the correct dependence on both the threshold and the
average load.

\subsection{Dense two-sided bounds}

There is another line of work on binary linear hashing in the dense regime,
where the number of keys is larger than the number of bins. This asks for a
two-sided guarantee: not only should no bucket be too large, but no bucket
should be too small.

Dhar and Dvir \cite{DD24} proved strong \(\ell_\infty\)-type guarantees for random linear
maps using connections to finite-field Kakeya and the polynomial method. In a
related dense setting, their result shows that random linear maps can distribute
a large set nearly as well as fully independent hashing, up to constant factors
in the relevant parameters.

Jaber, Kumar and Zuckerman \cite{JKZ25} also prove a dense two-sided theorem. In their
Section~6, they show that if the number of keys \(m\) is at least on the order
of \(n\log n\), then with high probability every bucket has load within constant
factors of the average load \(m/n\). More precisely, for every
\(0<\varepsilon<1/2\), there are constants \(C_1<C_2\), depending on
\(\varepsilon\), such that if \(m\ge C_1^{-1}n\log n\), then
\[
 \Pr_h\left[
 \forall y\in\F_2^\ell,\quad
 C_1\frac mn
 \le
 |h^{-1}(y)\cap S|
 \le
 C_2\frac mn
 \right]
 \ge
 1-\varepsilon.
\]
Thus in the dense regime all buckets are balanced up to constant factors.

The results in this note are complementary to these dense two-sided bounds.
Our base optimization improves the maximum-load tail in the sparse
large-load regime, where the relevant threshold is much larger than the average
load. It does not improve the dense two-sided theorem above, whose goal is to
control all buckets at the scale \(m/n\), including the lower tail. Similarly,
the fixed-bucket estimates proved here give sharp far upper-tail bounds for one
prescribed bucket, but they do not address the main lower-tail difficulty in the
dense two-sided problem.

\subsection{Organization}

The rest of the paper is organized as follows. Section~2 proves the fixed-bucket
tail bound. We first give the upper bound for one prescribed bucket, and then
show that the \(2^{-a^2}\) behavior is sharp by using a subspace construction.

Section~3 proves the expected maximum-load bound. We recall the
Jaber--Kumar--Zuckerman potential framework, optimize the base in the potential
argument, remove the surjectivity assumption, and integrate the resulting tail
bound. We also compare the fixed-bucket and potential methods.

The appendices contain the auxiliary and extended results. Appendix~A gives
self-contained proofs of the potential lemmas used in Section~3. Appendix~B
extends the fixed-bucket tail estimate to \(m\) keys and \(n\) bins. Appendix~C
extends the maximum-load bound to \(m\) keys and \(n\) bins in the regime where
the fully independent maximum-load scale is much larger than the average load.

\section{A fixed-bucket tail bound}

We begin with a fixed-bucket estimate, which is independent of the potential
method used later for the maximum load. The goal is to understand the load of
one prescribed bucket \(y\), rather than the maximum over all buckets. First we
prove a uniform upper bound showing that, for every fixed \(y\), the probability
that \(\Load_h(y)\) exceeds a dyadic threshold \(2^a\) is at most on the order
of \(2^{-a^2}\). The proof uses a simple counting idea: a heavy bucket must
contain many linearly independent tuples, while a random linear map sends any
fixed independent tuple to the prescribed bucket with probability \(2^{-ma}\).

We then show that this \(2^{-a^2}\) scale is sharp. The lower-bound example
takes the input set to be essentially an \(m\)-dimensional subspace. In that
case the bucket load is governed by the nullity of a random \(m\times m\)
binary matrix, whose distribution has exactly the same \(2^{-a^2}\) behavior.
In fact, the construction nearly matches the leading constant in the dyadic
upper bound as \(a\to\infty\).

This section is independent of the
potential argument used later.

\subsection{Upper bound for a prescribed bucket}
We first prove the upper bound for one fixed bucket. The idea is simple: if the
bucket contains many keys, then it must contain many ordered linearly independent
tuples. But any fixed independent tuple is sent to the prescribed bucket with
very small probability. We count such tuples and then apply Markov's inequality.

\begin{lemma}[Independent tuples in a set of distinct nonzero vectors]
\label{lem:independent-tuples}
Let \(A\subseteq \F_2^n\setminus\{0\}\) have size \(q\). Let
 $a=\left\lceil \log (q+1)\right\rceil$.
Then \(A\) contains at least
\[
 \prod_{j=0}^{a-1}(q+1-2^j)
\]
ordered linearly independent \(a\)-tuples.
\end{lemma}

\begin{proof}
Choose the tuple sequentially. Suppose that
\(v_1,\dots,v_j\in A\) have already been chosen and are linearly independent.
Their span contains \(2^j\) vectors, of which at most \(2^j-1\) are nonzero
vectors of \(A\). Hence the number of choices for \(v_{j+1}\in A\) outside
\(\Span(v_1,\dots,v_j)\) is at least
 $q-(2^j-1)=q+1-2^j$.
Since \(j<a\), this quantity is positive. Multiplying over
\(j=0,1,\dots,a-1\) proves the lemma.
\end{proof}

\begin{theorem}[Fixed-bucket tail]
\label{thm:fixed-bucket}
Let
 $U=\{u_1,\dots,u_{2^m}\}\subseteq \F_2^n\setminus\{0\}$
be a set of distinct nonzero vectors, and let
 $B:\F_2^n\to\F_2^m$
be a uniformly random linear map. Fix \(y\in\F_2^m\), and define
\[
 Z_y:=|\{i:Bu_i=y\}|.
\]
Then, for every integer \(r\ge0\), if
$a=\left\lceil\log (r+2)\right\rceil$,
we have
\[
 \Prb[Z_y>r]
 \le
 \left(\prod_{j=0}^{a-1}(r+2-2^j)\right)^{-1}.
\]
\end{theorem}

\begin{proof}
Let
 $q=r+1$, and
$a=\left\lceil\log (q+1)\right\rceil$.
Let \(\mathcal I_a\) be the set of all ordered \(a\)-tuples
\((i_1,\ldots,i_a)\) such that
\(u_{i_1},\dots,u_{i_a}\) are linearly independent. For a linear map \(B\), define
\[
 T_a(B)
 :=
 \sum_{(i_1,\dots,i_a)\in\mathcal I_a}
 \mathbf 1[Bu_{i_1}=\cdots=Bu_{i_a}=y].
\]
Thus \(T_a(B)\) is the number of ordered linearly independent \(a\)-tuples
\((u_{i_1},\dots,u_{i_a})\) such that
 $Bu_{i_1}=\cdots=Bu_{i_a}=y$.

For a fixed ordered linearly independent \(a\)-tuple, the random vectors
\(Bu_{i_1},\dots,Bu_{i_a}\) are independent and uniformly distributed in
\(\F_2^m\).\footnote{Indeed, after fixing bases of \(\F_2^n\) and
\(\F_2^m\), the map \(B\) is represented by a uniformly random
\(m\times n\) binary matrix. Since \(u_{i_1},\dots,u_{i_a}\) are linearly
independent, the random vectors \(Bu_{i_1},\dots,Bu_{i_a}\) are independent
uniform elements of \(\F_2^m\).}
Therefore
\[
 \Prb_B[Bu_{i_1}=\cdots=Bu_{i_a}=y]=2^{-ma}.
\]
Taking expectation over the random choice of \(B\), linearity of expectation gives
\[
\begin{aligned}
 \E_B[T_a(B)]
 &=
 \sum_{(i_1,\dots,i_a)\in\mathcal I_a}
 \Prb_B[Bu_{i_1}=\cdots=Bu_{i_a}=y]  \\
 &=
 |\mathcal I_a|\,2^{-ma}.
\end{aligned}
\]
Since \(|U|=2^m\), the total number of ordered \(a\)-tuples from \(U\) is
\((2^m)^a=2^{ma}\). Hence
 $|\mathcal I_a|\le 2^{ma}$,
and therefore
\[
 \E_B[T_a(B)]
 \le
 2^{ma}2^{-ma}
 =
 1.
\]

If \(Z_y\ge q\), then the set
 $A:=\{u_i:Bu_i=y\}$
has size at least \(q\). By Lemma~\ref{lem:independent-tuples}, it contains at
least
\[
 M(q):=\prod_{j=0}^{a-1}(q+1-2^j)
\]
ordered linearly independent \(a\)-tuples. Thus
\[
 Z_y\ge q
 \quad\Longrightarrow\quad
 T_a(B)\ge M(q).
\]
By Markov's inequality,
\[
 \Prb[Z_y\ge q]
 \le
 \Prb[T_a(B)\ge M(q)]
 \le
 \frac{\E_B[T_a(B)]}{M(q)}
 \le
 \frac1{M(q)}.
\]
Since \(q=r+1\), this is the desired bound.
\end{proof}

Let
\[
 \gamma:=\prod_{j=1}^{\infty}(1-2^{-j}).
\]
The preceding theorem has a clean form at dyadic thresholds.

\begin{corollary}[Dyadic fixed-bucket thresholds]
\label{cor:dyadic-fixed-bucket}
Under the assumptions of Theorem~\ref{thm:fixed-bucket}, for every integer
\(a\ge1\),
\[
 \Prb[Z_y>2^a-2]
 \le
 \gamma^{-1}2^{-a^2}.
\]
\end{corollary}

\begin{proof}
Apply Theorem~\ref{thm:fixed-bucket} with \(r=2^a-2\). Then
\[
 \prod_{j=0}^{a-1}(r+2-2^j)
=
 \prod_{j=0}^{a-1}(2^a-2^j)
 =
 2^{a^2}\prod_{j=0}^{a-1}(1-2^{j-a})
 =
 2^{a^2}\prod_{s=1}^{a}(1-2^{-s}).
\]
Since
 $\prod_{s=1}^{a}(1-2^{-s})\ge \gamma$,
the result follows.
\end{proof}

\subsection{A matching lower bound for the fixed-bucket tail}
We next show that the fixed-bucket upper bound is essentially best possible.
The example is based on a subspace. In this case, the load of the zero bucket is
controlled by the nullity of a random binary matrix, and the nullity distribution
has the same \(2^{-a^2}\) behavior as the upper bound.

\begin{proposition}[Near-sharpness of the fixed-bucket bound]
\label{prop:fixed-bucket-sharpness}
Let
 $\gamma:=\prod_{j=1}^{\infty}(1-2^{-j})$.
For every integer \(a\ge1\), there is a sequence of examples with
\(|U|=2^m\), \(m\to\infty\), such that, for a uniformly random linear map
 $B:\F_2^n\to\F_2^m$,
one has
\[
 \liminf_{m\to\infty}
 \Prb\left[|\{u\in U:Bu=0\}|>2^a-2\right]
 \ge
 \gamma^{-1}2^{-a^2}(1-2^{-a})^2.
\]
In particular, the upper bound
\[
 \Prb[Z_0>2^a-2]\le \gamma^{-1}2^{-a^2}
\]
from Corollary~\ref{cor:dyadic-fixed-bucket} is asymptotically tight in the
leading constant as \(a\to\infty\).
\end{proposition}

\begin{proof}
Let \(W<\F_2^n\) be an \(m\)-dimensional subspace and choose
\(v\notin W\). Let
 $U:=(W\setminus\{0\})\cup\{v\}$.
Then
 $|U|=(2^m-1)+1=2^m$,
and all elements of \(U\) are distinct and nonzero.

Let \(M\) be the restriction of \(B\) to \(W\):
\[
 M:=B|_W:W\to\F_2^m.
\]
Since \(\dim W=m\), after choosing a basis of \(W\), the map \(M\) is represented
by an \(m\times m\) binary matrix. As \(B\) is uniformly random, \(M\) is a
uniformly random linear map from \(W\) to \(\F_2^m\), equivalently a uniformly
random \(m\times m\) binary matrix.

If\footnote{Here \(\nul(M):=\dim\ker M=m-\rank(M)\), since \(\dim W=m\).} \(\nul(M)\ge a\), then
\[
 |(W\setminus\{0\})\cap\ker B|
 =
 2^{\nul(M)}-1
 \ge
 2^a-1.
\]
Therefore
\[
 |\{u\in U:Bu=0\}|>2^a-2.
\]
Consequently,
\[
 \Prb\left[|\{u\in U:Bu=0\}|>2^a-2\right]
 \ge
 \Prb[\nul(M)\ge a]
 \ge
 \Prb[\nul(M)=a].
\]

We use the standard rank distribution formula for random matrices over finite
fields; see, for example, Fulman and Goldstein~\cite{FulmanGoldstein2015}.
For a uniformly random \(m\times m\) matrix over \(\F_2\), the probability of
nullity exactly \(a\), equivalently rank \(m-a\), is
\[
 p_{m,a}
 =
 2^{-a^2}
 \frac{\prod_{j=a+1}^{m}(1-2^{-j})^2}
      {\prod_{j=1}^{m-a}(1-2^{-j})}.
\]
For fixed \(a\), as \(m\to\infty\),
\[
 p_{m,a}
 \longrightarrow
 2^{-a^2}
 \frac{\prod_{j=a+1}^{\infty}(1-2^{-j})^2}{\gamma}.
\]
Using the elementary inequality
 $\prod_j(1-x_j)\ge 1-\sum_j x_j$
for \(0\le x_j\le1\), we get
\[
 \prod_{j=a+1}^{\infty}(1-2^{-j})
 \ge
 1-\sum_{j=a+1}^{\infty}2^{-j}
 =
 1-2^{-a}.
\]
Hence
\[
 \liminf_{m\to\infty}
 \Prb\left[|\{u\in U:Bu=0\}|>2^a-2\right]
 \ge
 \gamma^{-1}2^{-a^2}(1-2^{-a})^2.
\]
Since \((1-2^{-a})^2\to1\) as \(a\to\infty\), this matches the leading constant
\(\gamma^{-1}\) in the dyadic fixed-bucket upper bound.
\end{proof}

\begin{remark}[Why the fixed-bucket bound is not enough for the expectation]
The fixed-bucket bound controls the load of one prescribed bucket. To control
the maximum load, one could try to apply it to all \(2^\ell\) buckets and then
take a union bound. This gives
\[
 \Prb[\M(S,h)>2^a-2]
 \le
 \gamma^{-1}2^{\ell-a^2}.
\]
This bound becomes useful only when \(a^2\) is at least comparable to \(\ell\).
Equivalently, it only controls very large loads, roughly of size
\(2^{\sqrt{\ell}}\) or larger.

However, the expected maximum load is much smaller: it is of order
 ${\ell}/{\log \ell}$.
The logarithm of this load is only \(O(\log\ell)\), far below
\(\sqrt{\ell}\). Therefore the fixed-bucket estimate is useful for the far tail,
but it does not by itself prove the correct expectation bound. For the
expectation-scale bound, we need the potential method used later.
\end{remark}

\section{The Expected Maximum-Load Bound}
In this section we prove the main maximum-load estimate. We first recall the
potential framework of Jaber--Kumar--Zuckerman for uniformly random surjective
linear maps. We then optimize the choice of the potential base to get a sharper
tail bound. Finally, we remove the surjectivity assumption and integrate the
tail to obtain the second-order expectation bound.

\subsection{The Jaber--Kumar--Zuckerman potential framework}

We now recall the potential framework used by Jaber, Kumar and Zuckerman. We
first work with uniformly random surjective maps. 

Let
 $H:\F_2^U\to\F_2^\ell$
be a uniformly random surjective linear map, and let
$ k:=U-\ell$.
A uniformly random surjective map can be sampled by first choosing a uniformly
random \(k\)-dimensional kernel
 $V\le\F_2^U$
and then choosing an isomorphism \(\F_2^U/V\cong\F_2^\ell\). We build the kernel gradually through a chain
\[
 V_0\le V_1\le\cdots\le V_k=V,
\]
where \(V_0=\{0\}\) and \(\dim V_i=i\). Thus, at step \(i\), the subspace
\(V_i\) is the \(i\)-dimensional part of the kernel that has been revealed.

Let \(S\subseteq\F_2^U\) have size \(2^\ell\). For \(x\in\F_2^U\), define
\[
 S_i(x):=|(x+V_i)\cap S|.
\]
The quantity \(S_i(x)\) is the load of the partial bucket \(x+V_i\). At
\(i=0\), the buckets are singletons, so \(S_0(x)=1_S(x)\). At the final step,
\(V_k=\ker H\), and the cosets \(x+V_k=x+\ker H\) are exactly the final buckets
of the hash map \(H\): all points in the same coset have the same image under
\(H\). Thus \(S_k(x)\) is the actual load of the final bucket containing \(x\).

Equivalently, let
\[
G_i:=\F_2^U/V_i
\]
be the quotient group of cosets of \(V_i\), and define
$f_i:G_i\to\mathbb Z_{\ge0}$
by
\[
f_i(x+V_i):=|(x+V_i)\cap S|.
\]
Then
\[
S_i(x)=f_i(x+V_i).
\]

For a base \(b>1\), define the exponential potential
\[
 \Phi_i:=\E_{x\in\F_2^U}\left[b^{S_i(x)}\right].
\]
Since all cosets of \(V_i\) have the same size and \(S_i\) is constant on each
coset, this is equivalently
\[
 \Phi_i=\E_{C\in G_i}\left[b^{f_i(C)}\right].
\]

We use the following two lemmas from the potential analysis of
Jaber, Kumar and Zuckerman. Lemma~\ref{lem:potential-evolution} packages JKZ Lemmas 3 and 4 in the
normalization used here. For completeness, see the proof in Appendix~\ref{ApA}.

\begin{lemma}[JKZ potential evolution]
\label{lem:potential-evolution}
For every \(b>1\), the potentials satisfy
\[
 \E[\Phi_{i+1}\mid \Phi_0,\dots,\Phi_i]
 \le
 \Phi_i^2
\]
and
\[
 \Phi_{i+1}-1\ge 2(\Phi_i-1)
\]
for every \(0\le i<k\).
\end{lemma}
The next lemma is a slightly stronger version of the quadratic tail lemma needed
for the argument. It says that if a nonnegative process grows at least
multiplicatively in the sense \(X_{i+1}-1\ge 2(X_i-1)\), but its conditional
expectation is at most quadratic, then the final value has a polynomial tail.
We include a self-contained proof in Appendix~\ref{ApA}.

\begin{lemma}[Strengthened JKZ quadratic potential tail lemma]
\label{lem:quadratic-tail}
Let \(X_0>1\) be deterministic, and let \(X_1,\dots,X_k\) be nonnegative random
variables satisfying
\[
 X_{i+1}-1\ge 2(X_i-1)
\]
and
\[
 \E[X_{i+1}\mid X_0,\dots,X_i]
 \le
 X_i^2
\]
for every \(0\le i<k\). If
 $\tau\ge 1+4(X_0-1)$,
then
\[
 \Prb\left[X_k\ge \tau^{2^k}\right]
 \le
 48\left(\frac{X_0-1}{\tau-1}\right)^2.
\]
\end{lemma}

For \(y\in\F_2^\ell\), the {\it fiber} of \(H\) over \(y\) is the set
\[
H^{-1}(y)=\{x\in\F_2^U:H(x)=y\}.
\]
Since \(H\) is surjective and \(\ker H=V_k\), every fiber is a coset of \(V_k\).

We also need the following elementary observation.

\begin{lemma}[Heavy bin implies large potential]
\label{lem:heavy-bin-potential}
If some fiber of the surjective map \(H\) contains at least \(T\) elements of
\(S\), then
\[
 \Phi_k\ge \frac{b^T}{2^\ell}.
\]
\end{lemma}

\begin{proof}
A fiber of \(H\) is a coset \(x+V_k\). If
 $|(x+V_k)\cap S|\ge T$,
then for every \(z\in x+V_k\),
\[
 S_k(z)=|(z+V_k)\cap S|=|(x+V_k)\cap S|\ge T.
\]
The coset has size \(|V_k|=2^k\). Therefore, its contribution to the average
\(\Phi_k\) is at least
\[
 \frac{2^k b^T}{2^U}
 =\frac{b^T}{2^{U-k}}
 =
 \frac{b^T}{2^\ell}.
\]
\end{proof}

Finally, since \(V_0=\{0\}\), we have
\[
 S_0(x)=1_S(x).
\]
Therefore
\[
 \Phi_0
 =
 \left(1-\frac{|S|}{2^U}\right)
 +
 \frac{|S|}{2^U}b.
\]
Using \(|S|=2^\ell\) and \(k=U-\ell\), this gives
\begin{equation}
\label{eq:initial-potential}
 \Phi_0-1
 =
 \frac{b-1}{2^k}
 \le
 \frac{b}{2^k}.
\end{equation}

\subsection{The optimized tail bound for surjective maps}
We first prove the optimized tail bound in the simpler setting where the linear
map is conditioned to be surjective. The proof follows the
Jaber--Kumar--Zuckerman potential argument, but chooses the base according to
the target load level. This smaller base lowers the initial potential and gives
the extra factor in the tail bound.

We now optimize the base in the potential argument.

\begin{proposition}[Surjective base-optimized tail]
\label{prop:surjective}
There exist absolute constants \(C_0>0\), \(D_0>0\), and \(\ell_0\) such that
the following holds. Let \(U\ge\ell\), let \(S\subseteq\F_2^U\) have size
\(2^\ell\), and let
 $H:\F_2^U\to\F_2^\ell$
be a uniformly random surjective linear map. Then, for every real \(R>1\)
satisfying
\[
 R\ell^{1-1/R}\ge D_0\ln\ell
\]
and every \(\ell\ge\ell_0\),
\[
 \Prb_H\left[
 \M(S,H)
 \ge
 R\frac{\ell}{\log\ell}
 \right]
 \le
 C_0\frac{(\ln\ell)^2}{R^2\ell^{2-2/R}}.
\]
\end{proposition}
\begin{proof}
Let
 $k:=U-\ell$,
 $L:={\ell}/{\log \ell}$, and
 $T:=RL=R\ell/\log \ell$.
Choose
\[
 \alpha:=\frac1R+\frac1{\ln\ell}
\]
and set
 $b:=\ell^\alpha$.
Then
$b=\ell^{1/R}\ell^{1/\ln\ell}=e\ell^{1/R}$.
If \(\M(S,H)\ge T\), Lemma~\ref{lem:heavy-bin-potential} gives
 $\Phi_k\ge {b^T}/{2^\ell}$.
Since
 $\log  b=\alpha\log \ell$
and
 $T=R{\ell}/{\log \ell}$,
we have
 $b^T=2^{\alpha R\ell}$.
Thus
\[
 \frac{b^T}{2^\ell}=2^{(\alpha R-1)\ell}.
\]
Let
\[
 A:=\alpha R-1=\frac{R}{\ln\ell}.
\]
Hence
\begin{equation}
\label{eq:heavy-implies-large-potential}
 \M(S,H)\ge T
 \quad\Longrightarrow\quad
 \Phi_k\ge 2^{A\ell}.
\end{equation}

Define
\[
 \tau:=1+\frac{(\ln 2)A\ell}{2^k}.
\]
Then, using \(1+x\le e^x\), with \(x=\tau-1\),
\[
 \tau^{2^k}
 \le
 \exp\left(2^k(\tau-1)\right)
 =
 \exp\left((\ln 2)A\ell\right)
 =
 2^{A\ell}.
\]
Therefore, by \eqref{eq:heavy-implies-large-potential},
\[
 \M(S,H)\ge T
 \quad\Longrightarrow\quad
 \Phi_k\ge \tau^{2^k}.
\]

Choose
 $D_0\ge {4e}/{\ln2}$.
Using \eqref{eq:initial-potential}, the definition of \(\tau\), and the
hypothesis \(R\ell^{1-1/R}\ge D_0\ln\ell\), we get
\begin{eqnarray*}
 4(\Phi_0-1)
 &\le&
 \frac{4b}{2^k}=
 \frac{4e\ell^{1/R}}{2^k}\\
 &\le&
 \frac{(\ln2)R\ell}{2^k\ln\ell}=
 \frac{(\ln2)A\ell}{2^k}=
 \tau-1.
\end{eqnarray*}
Thus
 $1+4(\Phi_0-1)\le \tau$.

Applying Lemma~\ref{lem:quadratic-tail} with \(X_i=\Phi_i\), we get
\[
 \Prb[\M(S,H)\ge T]
 \le
 \Prb[\Phi_k\ge \tau^{2^k}]
 \le
 48\left(\frac{\Phi_0-1}{\tau-1}\right)^2.
\]
Using \eqref{eq:initial-potential} and the definition of \(\tau\), and substituting
 $b=e\ell^{1/R}$,
and 
 $A={R}/{\ln\ell}$,
\[
 \frac{\Phi_0-1}{\tau-1}
 \le
 \frac{b/2^k}{(\ln 2)A\ell/2^k}
 =
 \frac{b}{(\ln 2)A\ell}=\frac{e\ln\ell}{(\ln 2)R\ell^{1-1/R}}.
\]
Therefore
\[
 \Prb[\M(S,H)\ge T]
 \le
 C_0
 \frac{(\ln\ell)^2}{R^2\ell^{2-2/R}}.
\]
This proves the claimed bound with
 $C_0=48\left({e}/{\ln2}\right)^2$.
\end{proof}

\subsection{Removing the surjectivity assumption}
We now pass from uniformly random surjective linear maps to uniformly random
linear maps. The idea is to embed the original space into a larger space. A
random map from the larger space is surjective with very high probability, and
its restriction to the original space is still a uniformly random linear map.

\begin{lemma}[Rank deficiency]
\label{lem:rank-deficiency}
Let
 $H:\F_2^U\to\F_2^\ell$
be a uniformly random linear map. Then
\[
 \Prb[H\text{ is not surjective}]
 \le
 2^{\ell-U}.
\]
\end{lemma}

\begin{proof}
After fixing bases of \(\F_2^U\) and \(\F_2^\ell\), the map \(H\) is represented
by a uniformly random \(\ell\times U\) binary matrix. The map \(H\) is not
surjective exactly when this matrix has rank less than \(\ell\).

The probability that a random \(\ell\times U\) binary matrix has full row rank
is
\[
 \prod_{j=0}^{\ell-1}(1-2^{j-U}).
\]
Therefore
\[
 \Prb[H\text{ is not surjective}]
 =
 1-\prod_{j=0}^{\ell-1}(1-2^{j-U})\le
 \sum_{j=0}^{\ell-1}2^{j-U}
 =
 2^{-U}(2^\ell-1)
 \le
 2^{\ell-U}.
\]
\end{proof}

\begin{theorem}[Base-optimized maximum-load tail]
\label{thm:main}
There exist absolute constants \(C>0\), \(D>0\), and \(\ell_0\) such that the
following holds. Let \(u\ge\ell\),
 $n:=2^\ell$,
and let \(S\subseteq\F_2^u\) have size \(n\). Let
 $h:\F_2^u\to\F_2^\ell$
be a uniformly random linear map. Then, for every real \(R>1\) satisfying
\[
 R\ell^{1-1/R}\ge D\ln\ell
\]
and every \(\ell\ge\ell_0\),
\[
 \Prb_h\left[
 \M(S,h)
 \ge
 R\frac{\ell}{\log \ell}
 \right]
 \le
 C\frac{(\ln\ell)^2}{R^2\ell^{2-2/R}}.
\]
Equivalently,
\[
 \Prb_h\left[
 \M(S,h)
 \ge
 R\frac{\log  n}{\log \log  n}
 \right]
 \le
 C
 \frac{(\log\log n)^2}
 {R^2(\log n)^{2-2/R}},
\]
for every \(R>1\) satisfying 
\[
 R(\log n)^{1-1/R}\ge D\log\log n,
\]
where changing the base of the logarithms only changes the absolute constants.
\end{theorem}

\begin{proof}
Let
\[
 p_{\ell,R}:=C_0\frac{(\ln\ell)^2}{R^2\ell^{2-2/R}},
\]
where \(C_0\) is the constant from Proposition~\ref{prop:surjective}. Choose
 $D\ge \max(D_0,\sqrt{C_0})$.
Then, whenever
 $R\ell^{1-1/R}\ge D\ln\ell$,
we have
\[
 p_{\ell,R}
 =
 C_0\frac{(\ln\ell)^2}{R^2\ell^{2-2/R}}
 =
 C_0\left(\frac{\ln\ell}{R\ell^{1-1/R}}\right)^2
 \le
 \frac{C_0}{D^2}
 \le
 1.
\]

Fix \(u\ge\ell\) and \(S\subseteq\F_2^u\) of size \(2^\ell\). Choose an integer
\(U\ge u\) so large that
 $2^{\ell-U}\le p_{\ell,R}$.
Embed \(\F_2^u\) into \(\F_2^U\) by appending \(U-u\) zero coordinates:
\[
 x=(x_1,\dots,x_u)\mapsto (x_1,\dots,x_u,0,\dots,0).
\]
Let \(V\le \F_2^U\) be the image of this embedding. We then regard
\(S\subseteq\F_2^u\) as a subset of \(V\subseteq\F_2^U\).

Let
 $H:\F_2^U\to\F_2^\ell$
be a uniformly random linear map. By Lemma~\ref{lem:rank-deficiency},
\[
 \Prb[H\text{ is not surjective}]
 \le 2^{\ell-U}\le 
 p_{\ell,R}.
\]
Conditioned on \(H\) being surjective, the map \(H\) is uniformly distributed
among all surjective maps \(\F_2^U\to\F_2^\ell\). Therefore
Proposition~\ref{prop:surjective} gives
\[
 \Prb\left[
 \M(S,H)
 \ge
 R\frac{\ell}{\log \ell}
 \;\middle|\;
 H\text{ is surjective}
 \right]
 \le
 p_{\ell,R}.
\]
Hence
\[
 \Prb_H\left[
 \M(S,H)
 \ge
 R\frac{\ell}{\log \ell}
 \right]
 \le
 2p_{\ell,R}.
\]

Finally, the restriction
 $H|_V:V\to\F_2^\ell$
is a uniformly random linear map from \(V\cong\F_2^u\) to \(\F_2^\ell\). Also,
for the set \(S\subseteq V\),
 $\M(S,H)=\M(S,H|_V)$,
because loads are computed only using points of \(S\). Thus the same bound
holds for a uniformly random linear map
 $h:\F_2^u\to\F_2^\ell$.
Absorbing the factor \(2\) into the absolute constant proves the theorem.
\end{proof}

\subsection{The expectation bound}

The optimized tail can be integrated starting slightly above
\(R=1+\ln\ln\ell/\ln\ell\). We choose the starting point with a fixed additive
constant in the numerator.

\begin{corollary}[Expectation bound to second order]
\label{cor:expectation}
For every \(u\ge\ell\) and every \(S\subseteq\F_2^u\) of size \(2^\ell\),
\[
 \E_h[\M(S,h)]
 \le
 \left(
  1+
  (1+o(1))\frac{\ln\ln\ell}{\ln\ell}
 \right)
 \frac{\ell}{\log\ell}
\]
as \(\ell\to\infty\), where \(h:\F_2^u\to\F_2^\ell\) is a uniformly random
linear map. Equivalently, since \(n=2^\ell\),
\[
 \E_h[\M(S,h)]
 \le
 \left(
  1+
  (1+o(1))\frac{\log\log\log n}{\log\log n}
 \right)
 \frac{\log n}{\log\log n},
\]
where the logarithms in the final display are base \(2\).
\end{corollary}

\begin{proof}
Let
\[
 L:=\frac{\ell}{\log\ell}.
\]
Let \(D\) be the constant from Theorem~\ref{thm:main}. Choose a fixed constant
\(F>0\) large enough so that
 $e^F>2D$.
Let \[
 R_0
 :=
 1+
 \frac{\ln\ln\ell+F}{\ln\ell}.
\]
We first check that Theorem~\ref{thm:main} applies for every \(R\ge R_0\), once
\(\ell\) is sufficiently large. We have
\[
 \left(1-\frac1{R_0}\right)\ln\ell
 =
 \frac{\ln\ln\ell+F}{R_0}
 =
 \ln\ln\ell+F+o(1).
\]
Therefore
\begin{eqnarray}\label{HIHO}
    \ell^{1-1/R_0}
 =
 \exp\left(
  \left(1-\frac1{R_0}\right)\ln\ell
 \right)
 =
 e^F(1+o(1))\ln\ell.
\end{eqnarray}
Hence, since \(R_0=1+o(1)\) and \(e^F>2D\),
\[
 R_0\ell^{1-1/R_0}
 =
 e^F(1+o(1))\ln\ell
 \ge
 D\ln\ell
\]
for all sufficiently large \(\ell\). Since the function
 $R\mapsto R\ell^{1-1/R}$
is increasing for \(R>0\), we also have
 $R\ell^{1-1/R}\ge D\ln\ell$
for every \(R\ge R_0\). Thus Theorem~\ref{thm:main} applies throughout the
range \(R\ge R_0\).

Using the tail-integral formula for a nonnegative random variable,
\begin{eqnarray*}
 \E[\M(S,h)]
 &=&
 \int_{0}^{\infty}
 \Prb[\M(S,h)\ge t]\,dt\\
 &=&
 \int_{0}^{R_0L}
 \Prb[\M(S,h)\ge t]\,dt
 +
 \int_{R_0L}^{\infty}
 \Prb[\M(S,h)\ge t]\,dt\\
 &\le&
 R_0L
 +
 \int_{R_0L}^{\infty}
 \Prb[\M(S,h)\ge t]\,dt\\
 &=&
 R_0L
 +
 L\int_{R_0}^{\infty}
 \Prb[\M(S,h)\ge RL]\,dR.
\end{eqnarray*}
By Theorem~\ref{thm:main},
\[
 \E[\M(S,h)]
 \le
 R_0L
 +
 CL\int_{R_0}^{\infty}
 \frac{(\ln\ell)^2}{R^2\ell^{2-2/R}}\,dR.
\]
Set
 $x=1-1/R$.
Then \(dx=dR/R^2\), and therefore
\[
 \int_{R_0}^{\infty}
 \frac{(\ln\ell)^2}{R^2\ell^{2-2/R}}\,dR
 =
 (\ln\ell)^2
 \int_{x_0}^{1}
 \ell^{-2x}\,dx,
\]
where
 $x_0=1-1/{R_0}$.
Thus, by (\ref{HIHO}),
\begin{eqnarray*}
 (\ln\ell)^2
 \int_{x_0}^{1}
 \ell^{-2x}\,dx
 &\le&
 \frac{\ln\ell}{2}\ell^{-2x_0}=
 \frac{\ln\ell}{2}\ell^{-2(1-1/R_0)}\\
 &=&
 \frac12 e^{-2F+o(1)}(\ln\ell)^{-1}\\
 &=&
 o\left(\frac{\ln\ln\ell}{\ln\ell}\right).
\end{eqnarray*}
It follows that
\[
 \E[\M(S,h)]
 \le
 \left(
  R_0+
  o\left(\frac{\ln\ln\ell}{\ln\ell}\right)
 \right)L.
\]
Substituting the definition of \(R_0\), we obtain
\[
 \E[\M(S,h)]
 \le
 \left(
  1+
  \frac{\ln\ln\ell+F}{\ln\ell}
  +
  o\left(\frac{\ln\ln\ell}{\ln\ell}\right)
 \right)
 \frac{\ell}{\log\ell}
=
 \left(
  1+
  (1+o(1))\frac{\ln\ln\ell}{\ln\ell}
 \right)
 \frac{\ell}{\log\ell}.
\]
Finally, since \(n=2^\ell\), we have \(\ell=\log n\) and
\(\log\ell=\log\log n\). Also,
\[
 \frac{\ln\ln\ell}{\ln\ell}
 =
 \left(1+o(1)\right)
 \frac{\log\log\log n}{\log\log n}.
\]
This gives the equivalent \(n\)-form.
\end{proof}

\subsection{Comparison of the two tail mechanisms}

The fixed-bucket and potential arguments apply in different regimes.
Corollary~\ref{cor:dyadic-fixed-bucket} gives, for a prescribed bucket,
\[
 \Prb[\Load_h(y)>2^a-2]
 \le
 \gamma^{-1}2^{-a^2}.
\]
This estimate is sharp for a single bucket. However, controlling the maximum
load by a direct union bound gives
\[
 \Prb[\M(S,h)>2^a-2]
 \le
 \gamma^{-1}2^{\ell-a^2}.
\]
This becomes nontrivial only when \(a\gtrsim\sqrt{\ell}\), that is, only for
loads at least about \(2^{\sqrt{\ell}}\). Thus the fixed-bucket estimate is a
sharp far-tail result, but it does not by itself reach the expectation scale.

By contrast, the optimized potential tail in Theorem~\ref{thm:main} controls
loads at the expectation scale
 ${\ell}/{\log \ell}$.
The improvement over the \(O(R^{-2})\) tail bound of Jaber--Kumar--Zuckerman
comes from reducing the potential base. To detect a bin of load
\[
 R\frac{\ell}{\log \ell},
\]
we take
\[
 b=\ell^{1/R+1/\ln\ell}=e\ell^{1/R}.
\]
The initial potential satisfies
\[
 \Phi_0-1=\frac{b-1}{2^k}.
\]
With the original choice \(b\approx\ell\), this is of order \(\ell/2^k\).
With the optimized choice \(b=e\ell^{1/R}\), it is instead of order
\(\ell^{1/R}/2^k\). Thus the initial potential is smaller by a factor
\[
 \ell^{-1+1/R}.
\]
The quadratic tail lemma bounds the failure probability by a constant times the
square of the normalized initial potential. Hence this saving is squared,
giving the factor
\[
 \ell^{-2+2/R}.
\]

The condition
\[
 R\ell^{1-1/R}\ge D\ln\ell
\]
is the condition that the chosen base is large enough to detect the threshold
while keeping the initial potential small enough for the quadratic lemma. The
threshold for this condition occurs when \(R\) is just above
\[
 1+\frac{\ln\ln\ell}{\ln\ell}.
\]
More precisely, if \(F\) is a sufficiently large absolute constant and
\[
 R_0=1+\frac{\ln\ln\ell+F}{\ln\ell},
\]
then
\[
 R_0\ell^{1-1/R_0}
 =
 e^F(1+o(1))\ln\ell
 \ge
 D\ln\ell
\]
for all sufficiently large \(\ell\). Hence Theorem~\ref{thm:main} applies for
all \(R\ge R_0\). Integrating the tail from this point gives
Corollary~\ref{cor:expectation}, namely
\[
 \E_h[\M(S,h)]
 \le
 \left(
 1+(1+o(1))\frac{\ln\ln\ell}{\ln\ell}
 \right)
 \frac{\ell}{\log\ell}.
\]

Thus the two estimates are complementary. The fixed-bucket argument gives a
nearly sharp tail for one prescribed bucket, including the correct
\(2^{-a^2}\) behavior. The potential argument is weaker for a single bucket but
is global: it controls the maximum over all buckets at the scale relevant for
the expected maximum load.

\section{Conclusion}

This note revisits the potential method of Jaber--Kumar--Zuckerman and shows
that the choice of base can be tuned to the target load level. This simple
optimization strengthens the maximum-load tail from the \(O(R^{-2})\) bound to
\[
 \Prb\left[
 \M(S,h)
 \ge
 R\frac{\ell}{\log \ell}
 \right]
 \le
 C\frac{(\ln\ell)^2}{R^2\ell^{2-2/R}}
\]
for every admissible \(R>1\). Integrating the estimate from
\[
 R_0=1+\frac{\ln\ln\ell+O(1)}{\ln\ell}
\]
gives
\[
 \E[\M(S,h)]
 \le
 \left(
 1+
 (1+o(1))
 \frac{\ln\ln\ell}{\ln\ell}
 \right)
 \frac{\ell}{\log \ell}.
\]
Equivalently, since \(n=2^\ell\),
\[
 \E[\M(S,h)]
 \le
 \frac{\log n}{\log\log n}
 +
 (1+o(1))
 \frac{\log n\cdot\log\log\log n}{(\log\log n)^2}.
\]
Thus binary linear hashing matches the fully independent maximum-load scale not
only in the leading term, but also in the dominant second-order term.

We also proved a separate fixed-bucket estimate. For one prescribed bucket,
the load has tail
\[
 \Prb[\Load_h(y)>2^a-2]\le \gamma^{-1}2^{-a^2},
\]
and a subspace construction shows that this is asymptotically tight even in the
leading constant as \(a\to\infty\). This fixed-bucket result is sharp, but it is
a far-tail statement: a direct union bound over all buckets loses a factor
\(2^\ell\), so the potential method remains essential for controlling the
maximum load at the expectation scale.

The same optimization also extends to \(m\) keys and \(n\) bins. In the regime
where the fully independent maximum-load scale \(t\) is much larger than the
average load \(m/n\), the argument gives
\[
 \E[\M(S,h)]\le (1+o(1))t.
\]
The dense two-sided balancing regime is different: there one must control both
large and small buckets at the scale of the average load. The methods here
improve the sparse maximum-load tail, but they do not improve the dense
two-sided bounds.

\newpage
\appendix

\section{Potential lemmas}\label{ApA}

In this appendix we include self-contained proofs of the potential lemmas used
in the main text. The first lemma packages the potential-evolution estimates of
Jaber--Kumar--Zuckerman in the normalization used here. The second lemma is a
normalized quadratic tail lemma.

\subsection{Potential evolution}

\begin{proof}[Proof of Lemma~\ref{lem:potential-evolution}]
Fix \(0\le i<k\), and condition on the subspace \(V_i\). Let
\[
 G_i:=\F_2^U/V_i,
 \qquad
 N_i:=|G_i|=2^{U-i}.
\]
Since \(V_{i+1}\) is obtained from \(V_i\) by adjoining one new vector outside
\(V_i\), let \(w\in G_i\setminus\{0\}\) denote the image of that new vector in
the quotient \(G_i=\F_2^U/V_i\). Conditional on \(V_i\), this \(w\) is uniformly
distributed over \(G_i\setminus\{0\}\).

Define
 $f_i:G_i\to\mathbb Z_{\ge0}$
by
\[
 f_i(C):=|C\cap S|.
\]
Thus
\[
 \Phi_i=\E_{C\in G_i}\left[b^{f_i(C)}\right].
\]

Let
\[
 \pi_i:G_i\to G_{i+1}
\]
be the quotient map induced by \(V_i\le V_{i+1}\). Passing from \(V_i\) to
\(V_{i+1}\) merges the two \(V_i\)-cosets \(C\) and \(C+w\) into the same
\(V_{i+1}\)-coset \(\pi_i(C)\). Define
 $f_{i+1}:G_{i+1}\to \mathbb Z_{\ge0}$
by
\[
 f_{i+1}(D):=|D\cap S|.
\]
Then, for every \(C\in G_i\),\footnote{Here we identify a coset with the corresponding subset of
\(\F_2^U\). Thus \(\pi_i(C)\), which is formally a \(V_{i+1}\)-coset, is the
union of the two \(V_i\)-cosets \(C\) and \(C+w\), namely
$\pi_i(C)=C\cup(C+w)$.
}
\[
 f_{i+1}(\pi_i(C))
 =
 |\pi_i(C)\cap S|
 =
 f_i(C)+f_i(C+w).
\]
Therefore, for the choice of new vector \(w\),
\[
 \Phi_{i+1}(w)
 =
 \E_{D\in G_{i+1}}\left[b^{f_{i+1}(D)}\right].
\]
Since every \(D\in G_{i+1}\) has exactly two preimages in \(G_i\), namely
\(C\) and \(C+w\), we may equivalently average over \(C\in G_i\). Hence
\[
 \Phi_{i+1}(w)
 =
 \E_{C\in G_i}
 \left[
 b^{f_i(C)+f_i(C+w)}
 \right].
\]

Let
 $a_C:=b^{f_i(C)}$.
Then
 $\Phi_i=\E_{C\in G_i}[a_C]$.
Averaging over the uniformly random choice of \(w\in G_i\setminus\{0\}\), we get
\[
 \E_w[\Phi_{i+1}(w)]
 =
 \E_{w\neq0}\E_{C\in G_i}[a_Ca_{C+w}].
\]
Expanding this average gives
\[
 \E_{w\neq0}\E_{C\in G_i}[a_Ca_{C+w}]
 =
 \frac1{N_i(N_i-1)}
 \sum_{C\in G_i}
 \sum_{\substack{w\in G_i\\w\neq0}}
 a_Ca_{C+w}.
\]
For fixed \(C\), as \(w\) ranges over \(G_i\setminus\{0\}\), the coset \(C+w\)
ranges over \(G_i\setminus\{C\}\). Therefore
\[
 \sum_{C\in G_i}
 \sum_{\substack{w\in G_i\\w\neq0}}
 a_Ca_{C+w}
 =
 \sum_{C\in G_i}a_C
 \sum_{\substack{D\in G_i\\D\neq C}}a_D
 =
 \left(\sum_{C\in G_i}a_C\right)^2
 -
 \sum_{C\in G_i}a_C^2.
\]
Hence
\[
\begin{aligned}
 \E_w[\Phi_{i+1}(w)]
 &=
 \frac{1}{N_i(N_i-1)}
 \left[
 \left(\sum_{C\in G_i}a_C\right)^2
 -
 \sum_{C\in G_i}a_C^2
 \right]  \\
 &=
 \frac{1}{N_i(N_i-1)}
 \left[
 N_i^2\Phi_i^2
 -
 N_i\E_{C\in G_i}[a_C^2]
 \right] \\
 &=
 \frac{N_i\Phi_i^2-\E_{C\in G_i}[a_C^2]}{N_i-1}.
\end{aligned}
\]
By Jensen's inequality,
 $\E_{C\in G_i}[a_C^2]
 \ge
 \left(\E_{C\in G_i}[a_C]\right)^2
 =
 \Phi_i^2$.
Therefore
\[
 \E_w[\Phi_{i+1}(w)]
 \le
 \frac{N_i\Phi_i^2-\Phi_i^2}{N_i-1}
 =
 \Phi_i^2.
\]
Equivalently,
 $\E[\Phi_{i+1}\mid V_i]\le \Phi_i^2$.
Since \(\Phi_0,\dots,\Phi_i\) are determined once \(V_i\) is known, the tower
property gives
\[
 \E[\Phi_{i+1}\mid \Phi_0,\dots,\Phi_i]
 =
 \E\left[
  \E[\Phi_{i+1}\mid V_i]
  \,\middle|\,
  \Phi_0,\dots,\Phi_i
 \right]
 \le
 \Phi_i^2.
\]

It remains to prove the deterministic growth estimate. For every \(C\in G_i\),
since \(b>1\), both
\[
 b^{f_i(C)}-1
 \qquad\text{and}\qquad
 b^{f_i(C+w)}-1
\]
are nonnegative. Therefore
\[
 b^{f_i(C)+f_i(C+w)}-1
 =
 b^{f_i(C)}b^{f_i(C+w)}-1
 \ge
 \bigl(b^{f_i(C)}-1\bigr)
 +
 \bigl(b^{f_i(C+w)}-1\bigr).
\]
Averaging over \(C\in G_i\), we obtain
\[
 \Phi_{i+1}-1
 \ge
 \E_{C\in G_i}\bigl[b^{f_i(C)}-1\bigr]
 +
 \E_{C\in G_i}\bigl[b^{f_i(C+w)}-1\bigr].
\]
Since \(C+w\) is uniformly distributed over \(G_i\) when \(C\) is uniformly
distributed over \(G_i\), the two averages are equal. Hence
\[
 \Phi_{i+1}-1
 \ge
 2\E_{C\in G_i}\bigl[b^{f_i(C)}-1\bigr]
 =
 2(\Phi_i-1).
\]
This completes the proof.
\end{proof}

\subsection{Strengthened quadratic potential tail lemma}

We next prove Lemma~\ref{lem:quadratic-tail}. This is the strengthened
quadratic tail lemma used in the main text; in particular, it does not require
any upper bound on \(\tau\).

We need one elementary one-step estimate.

\begin{lemma}[One-step estimate]
\label{lem:one-step-estimate}
Let \(s>0\), \(0\le r\le1/4\), and let
 $d:=s(2+s)$.

Let \(Z\) be a nonnegative random variable satisfying
 $Z\ge 2rs$
and
 $\E [Z]\le 2rs+r^2s^2$.
Define
 $F(u):=\min\{1,48u^2\}$.
Then
\[
 \E \left[F\left(\frac{Z}{d}\right)\right]\le 48r^2.
\]
\end{lemma}
\begin{proof}
The function
 $F(u)=\min\{1,48u^2\}$
is increasing and globally Lipschitz on \([0,\infty)\) with Lipschitz constant
\(8\sqrt3\). Indeed, on the quadratic part its derivative is \(96u\), and the
quadratic part ends at \(u=1/\sqrt{48}\), where the derivative equals
\(8\sqrt3\).

Since \(Z\ge 2rs\), we have
\[
 F\left(\frac{Z}{d}\right)
 \le
 F\left(\frac{2rs}{d}\right)
 +
 \frac{8\sqrt3}{d}(Z-2rs).
\]
Taking expectations and using
 $\E(Z-2rs)\le r^2s^2$,
we get
\[
 \E \left[F\left(\frac{Z}{d}\right)\right]
 \le
 F\left(\frac{2rs}{s(2+s)}\right)
 +
 \frac{8\sqrt3}{s(2+s)}r^2s^2
 =
 F\left(\frac{2r}{2+s}\right)
 +
 \frac{8\sqrt3\,r^2s}{2+s}.
\]

We now show that the right-hand side is at most \(48r^2\).

First suppose
\[
 \frac{2r}{2+s}\le \frac1{\sqrt{48}}.
\]
Then
\[
 F\left(\frac{2r}{2+s}\right)
 =
 48\left(\frac{2r}{2+s}\right)^2.
\]
It is enough to prove
\[
 48\frac{4r^2}{(2+s)^2}
 +
 \frac{8\sqrt3\,r^2s}{2+s}
 \le
 48r^2.
\]
After dividing by \(48r^2\), this becomes
\[
 \frac4{(2+s)^2}
 +
 \frac{\sqrt3}{6}\frac{s}{2+s}
 \le
 1.
\]
Since
\[
 1-\frac4{(2+s)^2}
 =
 \frac{(2+s)^2-4}{(2+s)^2}
 =
 \frac{s(4+s)}{(2+s)^2},
\]
it is enough to show
\[
 \frac{s(4+s)}{(2+s)^2}
 \ge
 \frac{\sqrt3}{6}\frac{s}{2+s}.
\]
Since \(s>0\), this is equivalent to
\[
 \frac{4+s}{2+s}
 \ge
 \frac{\sqrt3}{6}.
\]
Since
 ${(4+s)}/{(2+s)}>1>{\sqrt3}/{6}$,
the desired inequality follows.

Now suppose
\[
 \frac{2r}{2+s}> \frac1{\sqrt{48}}.
\]
Then
\[
 F\left(\frac{2r}{2+s}\right)=1.
\]
The assumption implies
\[
 48r^2-1>\frac{(2+s)^2}{4}-1
 =
 s+\frac{s^2}{4}.
\]
Since \(r\le 1/4\), we also have
\[
 s+\frac{s^2}{4}
 \ge
 \frac{\sqrt3}{4}s
 \ge
 \frac{\sqrt3}{2}\frac{s}{2+s}
 \ge
 \frac{8\sqrt3\,r^2s}{2+s}.
\]
Therefore
\[
 1+\frac{8\sqrt3\,r^2s}{2+s}
 \le
 48r^2.
\]
This proves the estimate in both cases.
\end{proof}

We now prove the strengthened quadratic tail lemma.

\begin{proof}[Proof of Lemma~~\ref{lem:quadratic-tail}]
We prove the lemma by induction on \(k\). Write
 $x:=X_0$.
The case \(k=0\) is immediate. Indeed, $\tau\ge 1+4(x-1)=x+3(x-1)>x$
and so
 $\Prb[X_0\ge\tau]=0$.

Assume the result is known for sequences of length \(k-1\), and prove it for
length \(k\). Let
$y=x-1$,
and
 $s=\tau-1$.
Then
 $0< s$,
 and
 $0\le y=x-1\le (\tau-1)/4={s}/{4}$.
Define
 $r:={y}/{s}$.
Thus \(0\le r=(x-1)/(\tau-1)\le1/4\). Also let
 $d:=\tau^2-1=s(2+s)$.

Let
\[
 Z:=X_1-1.
\]
Since $X_1-1\ge 2(X_0-1)$,
\[
 Z\ge 2y=2rs.
\]
The conditional expectation hypothesis gives
\[
 \E [Z]
 =
 \E[X_1-1]
 \le
 x^2-1
 =
 (1+y)^2-1
 =
 2y+y^2
 =
 2rs+r^2s^2.
\]

Now condition on \(X_1\). If
\[
 1+4(X_1-1)\le \tau^2,
\]
then the induction hypothesis, applied to the remaining sequence
\(X_1,X_2,\dots,X_k\) with threshold \(\tau^2\), gives
\[
 \Prb\left[
 X_k\ge (\tau^2)^{2^{k-1}}
 \,\middle|\,
 X_1
 \right]
 \le
 48\left(\frac{X_1-1}{\tau^2-1}\right)^2.
\]
If instead
\[
 1+4(X_1-1)> \tau^2,
\]
then, $48((X_1-1)/(\tau^2-1))^2>48/16>1$, and we use the trivial bound by \(1\). We have, in all cases,
\[
 \Prb\left[
 X_k\ge \tau^{2^k}
 \,\middle|\,
 X_1
 \right]
 \le
 F\left(\frac{X_1-1}{\tau^2-1}\right),
\]
where
\[
 F(u):=\min\{1,48u^2\}.
\]
Taking expectations,
\[
 \Prb[X_k\ge \tau^{2^k}]
 \le
 \E \left[F\left(\frac{Z}{d}\right)\right].
\]
By Lemma~\ref{lem:one-step-estimate},
\[
 \E \left[F\left(\frac{Z}{d}\right)\right]
 \le
 48r^2
 =
 48\left(\frac{y}{s}\right)^2
 =
 48\left(\frac{x-1}{\tau-1}\right)^2=
 48\left(\frac{X_0-1}{\tau-1}\right)^2.
\]
This completes the induction.
\end{proof}

\section{A fixed-bucket tail bound for \(m\) keys and \(n\) bins}
\label{app:fixed-bucket-m-n}

In this appendix we record the fixed-bucket version of the tail estimate for
\(m\) keys and \(n\) bins. Since binary linear hashing naturally gives a
power-of-two number of bins, write
 $n:=2^\ell$.
Let
\[
 S=\{u_1,\dots,u_m\}\subseteq \F_2^d\setminus\{0\}
\]
be a set of distinct nonzero vectors, and let
 $h:\F_2^d\to \F_2^\ell$
be a uniformly random linear map. For a fixed bucket \(y\in\F_2^\ell\), define
\[
 Z_y:=|\{i:h(u_i)=y\}|.
\]
Let
\[
 \lambda:=\frac mn.
\]

We will use again Lemma~\ref{lem:independent-tuples}: if
\(A\subseteq\F_2^d\setminus\{0\}\) has size \(q\), and
 $a=\left\lceil\log(q+1)\right\rceil$,
then \(A\) contains at least
\[
 \prod_{j=0}^{a-1}(q+1-2^j)
\]
ordered linearly independent \(a\)-tuples. 

\begin{theorem}[Fixed-bucket tail for \(m\) keys and \(n\) bins]
\label{thm:fixed-bucket-m-n}
Let \(n=2^\ell\), and let
\[
 S=\{u_1,\dots,u_m\}\subseteq \F_2^d\setminus\{0\}
\]
be a set of distinct nonzero vectors. Let
 $h:\F_2^d\to \F_2^\ell$
be a uniformly random linear map. Fix \(y\in\F_2^\ell\), and define
 $Z_y:=|\{i:h(u_i)=y\}|$.
Let
 $\lambda= m/n$.
Then, for every integer \(r\ge0\), if
\[
 a:=\left\lceil \log (r+2)\right\rceil,
\]
we have
\[
 \Prb[Z_y>r]
 \le
 \lambda^a
 \left(\prod_{j=0}^{a-1}(r+2-2^j)\right)^{-1}.
\]
\end{theorem}

\begin{proof}
Let
 $q=r+1$,
and
 $a=\left\lceil \log (q+1)\right\rceil$.
Let \(\mathcal I_a\) be the set of ordered \(a\)-tuples
 $(i_1,\dots,i_a)$
such that
 $u_{i_1},\dots,u_{i_a}$
are linearly independent. For a linear map \(h\), define
\[
 T_a(h)
 :=
 \sum_{(i_1,\dots,i_a)\in\mathcal I_a}
 \mathbf 1[h(u_{i_1})=\cdots=h(u_{i_a})=y].
\]
Thus \(T_a(h)\) counts ordered linearly independent \(a\)-tuples from \(S\)
that all land in the prescribed bucket \(y\).

For a fixed ordered linearly independent \(a\)-tuple
\((u_{i_1},\dots,u_{i_a})\), the random vectors
 $h(u_{i_1}),\dots$ $,h(u_{i_a})$
are independent and uniformly distributed in \(\F_2^\ell\). Hence
\[
 \Prb_h[h(u_{i_1})=\cdots=h(u_{i_a})=y]
 =
 2^{-\ell a}
 =
 n^{-a}.
\]
Therefore, by linearity of expectation,
\[
 \E_h[T_a(h)]
 =
 |\mathcal I_a|n^{-a}.
\]
Since
 $|\mathcal I_a|\le m^a$,
we get
\[
 \E_h[T_a(h)]
 \le
 \left(\frac mn\right)^a
 =
 \lambda^a.
\]

Now suppose \(Z_y\ge q\). Then the set
 $A=\{u_i:h(u_i)=y\}$
has size at least \(q\). By Lemma~\ref{lem:independent-tuples}, \(A\)
contains at least
\[
 M(q):=\prod_{j=0}^{a-1}(q+1-2^j)
\]
ordered linearly independent \(a\)-tuples. Thus
\[
 Z_y\ge q
 \quad\Longrightarrow\quad
 T_a(h)\ge M(q).
\]
By Markov's inequality,
\[
 \Prb[Z_y\ge q]
 \le
 \Prb[T_a(h)\ge M(q)]
 \le
 \frac{\E_h[T_a(h)]}{M(q)}
 \le
 \frac{\lambda^a}{M(q)}.
\]
Since \(q=r+1\), this gives
\[
 \Prb[Z_y>r]
 \le
 \lambda^a
 \left(\prod_{j=0}^{a-1}(r+2-2^j)\right)^{-1}.
\]
\end{proof}

Let
\[
 \gamma:=\prod_{j=1}^{\infty}(1-2^{-j}).
\]
The preceding theorem has the following clean form at dyadic thresholds.

\begin{corollary}[Dyadic fixed-bucket tail for \(m\) keys and \(n\) bins]
\label{cor:dyadic-fixed-bucket-m-n}
Under the assumptions of Theorem~\ref{thm:fixed-bucket-m-n}, for every integer
\(a\ge1\),
\[
 \Prb[Z_y>2^a-2]
 \le
 \gamma^{-1}\lambda^a2^{-a^2}.
\]
\end{corollary}

\begin{proof}
Apply Theorem~\ref{thm:fixed-bucket-m-n} with \(r=2^a-2\). Then
 $\left\lceil \log (r+2)\right\rceil=a$,
and
\[
\begin{aligned}
 \prod_{j=0}^{a-1}(r+2-2^j)
 &=
 \prod_{j=0}^{a-1}(2^a-2^j)  \\
 &\ge
 \gamma 2^{a^2}.
\end{aligned}
\]
Therefore
\[
 \Prb[Z_y>2^a-2]
 \le
 \lambda^a\gamma^{-1}2^{-a^2}.
\]
\end{proof}

\begin{remark}[Including the zero vector]
The theorem above assumes that all keys are nonzero. If \(0\in S\), then for
\(y\neq0\), the zero vector never contributes to \(Z_y\), so one may apply the
theorem to \(S\setminus\{0\}\). For \(y=0\), the zero vector contributes
deterministically one point to the load. Thus the same bound applies to
\(Z_0-1\) after removing the zero vector from \(S\).
\end{remark}

\begin{remark}[Comparison with the balanced case]
When \(m=n\), we have \(\lambda=1\), and Corollary~\ref{cor:dyadic-fixed-bucket-m-n}
becomes
\[
 \Prb[Z_y>2^a-2]\le \gamma^{-1}2^{-a^2},
\]
which is the fixed-bucket estimate used in the balanced case.
For general \(m\), the additional factor
\[
 \lambda^a=\left(\frac mn\right)^a
\]
reflects the average load of the prescribed bucket.
\end{remark}

\begin{proposition}[Matching lower bound for the fixed-bucket tail]
\label{prop:fixed-bucket-m-n-lower}
Let \(n=2^\ell\), let \(m=2^d\), and let
 $\lambda=m/n=2^{d-\ell}$.
Let
 $d\ge a\ge \max\{1,d-\ell\}$.
Then there exists a set
 $S\subseteq \F_2^N\setminus\{0\}$
of \(m\) distinct nonzero vectors such that, for a uniformly random linear map
 $h:\F_2^N\to \F_2^\ell$,
one has
\[
 \Prb\left[
 |\{x\in S:h(x)=0\}|>2^a-2
 \right]
 \ge
 \gamma^2\lambda^a2^{-a^2},
\]
where
 $\gamma:=\prod_{j=1}^{\infty}(1-2^{-j})$.

Thus, the upper bound
\[
 \Prb[Z_0>2^a-2]\le \gamma^{-1}\lambda^a2^{-a^2}
\]
is sharp up to an absolute multiplicative constant.
\end{proposition}

\begin{proof}
Let \(W\le \F_2^N\) be a \(d\)-dimensional subspace, and choose
\(v\notin W\). Define
 $S:=(W\setminus\{0\})\cup\{v\}$.
Then
 $|S|=(2^d-1)+1=2^d=m$,
and all elements of \(S\) are distinct and nonzero.

Let
\[
 M:=h|_W:W\to \F_2^\ell.
\]
After choosing a basis of \(W\), the map \(M\) is represented by a uniformly
random \(\ell\times d\) binary matrix. If
 $\nul(M)\ge a$,
then
\[
 |\ker(M)\setminus\{0\}|
 =
 2^{\nul(M)}-1
 \ge
 2^a-1.
\]
Since
\[
 \ker(M)\setminus\{0\}
 \subseteq
 \{x\in S:h(x)=0\},
\]
we get
\[
 |\{x\in S:h(x)=0\}|>2^a-2.
\]
Therefore
\[
 \Prb\left[
 |\{x\in S:h(x)=0\}|>2^a-2
 \right]
 \ge
 \Prb[\nul(M)\ge a]
 \ge
 \Prb[\nul(M)=a].
\]

It remains to lower-bound the probability that a uniformly random
\(\ell\times d\) binary matrix has nullity exactly \(a\). This is the same as
having rank \(d-a\). The standard rank formula in \cite{FulmanGoldstein2015} gives
\[
 \Prb[\rank(M)=d-a]
 =
 2^{-a(\ell-d+a)}
 \frac{
 \prod_{i=0}^{d-a-1}(1-2^{i-\ell})
 \prod_{i=0}^{d-a-1}(1-2^{i-d})
 }
 {
 \prod_{i=0}^{d-a-1}(1-2^{i-(d-a)})
 }.
\]
The denominator is at most \(1\). Also,
\[
 \prod_{i=0}^{d-a-1}(1-2^{i-d})
 =
 \prod_{j=a+1}^{d}(1-2^{-j})
 \ge
 \gamma,
\]
and, since \(a\ge d-\ell\), we have \(d-a\le \ell\), so
\[
 \prod_{i=0}^{d-a-1}(1-2^{i-\ell})
 \ge
 \gamma.
\]
Hence
\[
 \Prb[\nul(M)=a]
 =
 \Prb[\rank(M)=d-a]
 \ge
 \gamma^2 2^{-a(\ell-d+a)}.
\]
Finally,
\[
 2^{-a(\ell-d+a)}
 =
 2^{a(d-\ell)}2^{-a^2}
 =
 \lambda^a2^{-a^2}.
\]
Therefore
\[
 \Prb\left[
 |\{x\in S:h(x)=0\}|>2^a-2
 \right]
 \ge
 \gamma^2\lambda^a2^{-a^2}.
\]
\end{proof}

\section{Maximum-load bounds for \(m\) keys and \(n\) bins}
\label{app:m-keys-n-bins}

In this section we record the extension of the base-optimized argument to the
case where the number of keys is not necessarily equal to the number of bins.
Let the number of bins be
 $n=2^\ell$,
let \(S\subseteq\F_2^u\) have size
 $|S|=m$,
and write
\[
 \lambda:=\frac{m}{n}.
\]
Thus \(\lambda\) is the average load. For a linear map
\(h:\F_2^u\to\F_2^\ell\), define
\[
 \M(S,h):=\max_{y\in\F_2^\ell}|h^{-1}(y)\cap S|.
\]

The fully independent comparison scale in the sparse large-load regime is the
number \(t=t(m,n)\) defined by
\[
 t\ln\left(\frac{t}{e\lambda}\right)=\ln n.
\]
This scale is meaningful in the range
\[
 \frac{t}{\lambda}\to\infty,
\]
that is, when the maximum-load scale is much larger than the average load.

We first prove the tail bound for uniformly random surjective maps.

\begin{proposition}[Surjective tail bound for \(m\) keys and \(n\) bins]
\label{prop:general-surjective}
There exist absolute constants \(C_0,D_0>0\) such that the following holds.
Let \(U\ge \ell\), let \(S\subseteq\F_2^U\) have size \(m\), and let
 $H:\F_2^U\to\F_2^\ell$
be a uniformly random surjective linear map. Put \(n=2^\ell\) and
\(\lambda=m/n\). Then, for every \(T>0\) satisfying
\[
 T\ge D_0\lambda n^{1/T},
\]
one has
\[
 \Prb_H[\M(S,H)\ge T]
 \le
 C_0\left(\frac{\lambda n^{1/T}}{T}\right)^2.
\]
\end{proposition}

\begin{proof}
Let
 $k=U-\ell$.
As before, expose the kernel of \(H\) through a chain
\[
 V_0\le V_1\le\cdots\le V_k=\ker H,
\]
with \(V_0=\{0\}\) and \(\dim V_i=i\). For a base \(b>1\), define
\[
 \Phi_i:=\E_{x\in\F_2^U} \left[b^{S_i(x)}\right],
 \qquad
 S_i(x):=|(x+V_i)\cap S|.
\]
We use the same potential evolution and quadratic tail lemmas as before.

Choose
 $b=e n^{1/T}$.
Since \(V_0=\{0\}\), we have \(S_0(x)=1_S(x)\), and therefore
\[
 \Phi_0-1
 =
 \frac{|S|}{2^U}(b-1)
 =
 \frac{m}{n2^k}(b-1)
 \le
 \frac{\lambda b}{2^k}.
\]

If some final bucket has load at least \(T\), then, by Lemma~\ref{lem:heavy-bin-potential},
\[
 \Phi_k\ge \frac{b^T}{n}=
 \frac{(e n^{1/T})^T}{n}
 =
 e^T.
\]

Now define
\[
 \tau:=1+\frac{T}{2^k}.
\]
Then
\[
 \tau^{2^k}
 =
 \left(1+\frac{T}{2^k}\right)^{2^k}
 \le
 e^T.
\]
Hence
\[
 \M(S,H)\ge T
 \quad\Longrightarrow\quad
 \Phi_k\ge \tau^{2^k}.
\]

We verify the hypothesis of the quadratic potential tail lemma. Choose
 $D_0\ge 4e$.
Using \(\Phi_0-1\le \lambda b/2^k\), the definition of \(\tau\), and the
assumption \(T\ge D_0\lambda n^{1/T}\), we get
\begin{eqnarray*}
 4(\Phi_0-1)
 \le
 \frac{4\lambda b}{2^k}
 =
 \frac{4e\lambda n^{1/T}}{2^k}
 \le
 \frac{T}{2^k}
 =
 \tau-1.
\end{eqnarray*}
Thus
\[
 \tau\ge 1+4(\Phi_0-1).
\]
Applying the quadratic potential tail lemma with \(X_i=\Phi_i\), we get
\[
 \Prb_H[\M(S,H)\ge T]
 \le
 \Prb[\Phi_k\ge \tau^{2^k}]
 \le
 48\left(\frac{\Phi_0-1}{\tau-1}\right)^2.
\]
Finally,
\[
 \frac{\Phi_0-1}{\tau-1}
 \le
 \frac{\lambda b/2^k}{T/2^k}
 =
 \frac{\lambda b}{T}
 =
 \frac{e\lambda n^{1/T}}{T}.
\]
Therefore
\[
 \Prb_H[\M(S,H)\ge T]
 \le
 48e^2
 \left(\frac{\lambda n^{1/T}}{T}\right)^2.
\]
This proves the result with \(C_0=48e^2\).
\end{proof}

We next remove the surjectivity assumption, exactly as in the equal-size case.

\begin{theorem}[Tail bound for \(m\) keys and \(n\) bins]
\label{thm:general-tail}
There exist absolute constants \(C,D>0\) such that the following holds. Let
\(n=2^\ell\), let \(S\subseteq\F_2^u\) have size \(m\), and let
 $\lambda:=m/n$.
Let
 $h:\F_2^u\to\F_2^\ell$
be a uniformly random linear map. Then, for every \(T>0\) satisfying
 $T\ge D\lambda n^{1/T}$,
one has
\[
 \Prb_h[\M(S,h)\ge T]
 \le
 C\left(\frac{\lambda n^{1/T}}{T}\right)^2.
\]
\end{theorem}

\begin{proof}
Let
\[
 p_T:=C_0\left(\frac{\lambda n^{1/T}}{T}\right)^2,
\]
where \(C_0\) is the constant from Proposition~\ref{prop:general-surjective}.
Choose the constant \(D\) large enough so that \(p_T\le1\) whenever
\(T\ge D\lambda n^{1/T}\).

Fix \(u\ge\ell\) and \(S\subseteq\F_2^u\) of size \(m\). Choose \(U\ge u\) so
large that
 $2^{\ell-U}\le p_T$.
Embed \(\F_2^u\) into \(\F_2^U\) by appending \(U-u\) zero coordinates, and
regard \(S\) as a subset of this copy of \(\F_2^u\) inside \(\F_2^U\).

Let
 $H:\F_2^U\to\F_2^\ell$
be a uniformly random linear map. By the Lemma~\ref{lem:rank-deficiency},
\[
 \Prb[H\text{ is not surjective}]
 \le
 2^{\ell-U}
 \le
 p_T.
\]
Conditioned on \(H\) being surjective, the map \(H\) is uniformly distributed
among all surjective maps \(\F_2^U\to\F_2^\ell\). Hence
Proposition~\ref{prop:general-surjective} gives
\[
 \Prb\left[
 \M(S,H)\ge T
 \;\middle|\;
 H\text{ is surjective}
 \right]
 \le
 p_T.
\]
Therefore
\[
 \Prb_H[\M(S,H)\ge T]\le 2p_T.
\]

Finally, the restriction of \(H\) to the embedded copy of \(\F_2^u\) is a
uniformly random linear map \(\F_2^u\to\F_2^\ell\), and the loads are computed
only using points of \(S\). Thus the same estimate holds for a uniformly random
linear map \(h:\F_2^u\to\F_2^\ell\). Absorbing the factor \(2\) into the
constant proves the theorem.
\end{proof}

The previous theorem can be written in terms of the fully independent
large-load scale \(t\).

\begin{corollary}[Comparison with the fully independent scale]
\label{cor:general-comparison}
Let \(n=2^\ell\), let \(S\subseteq\F_2^u\) have size \(m\), and let
 $\lambda:= m/n$.
Let \(t=t(m,n)\) be defined by
\[
 t\ln\left(\frac{t}{e\lambda}\right)=\ln n.
\]
Then, for every \(R>1\) satisfying
\[
 R\left(\frac{t}{\lambda}\right)^{1-1/R}\ge D,
\]
one has
\[
 \Prb_h[\M(S,h)\ge Rt]
 \le
 \frac{C}{R^2}
 \left(\frac{\lambda}{t}\right)^{2-2/R}.
\]
\end{corollary}

\begin{proof}
Since
\[
 t\ln\left(\frac{t}{e\lambda}\right)=\ln n,
\]
we have
\[
 n^{1/t}=\frac{t}{e\lambda}.
\]
Apply Theorem~\ref{thm:general-tail} with
 $T:=Rt$.
Then
\[
 n^{1/T}
 =
 n^{1/(Rt)}
 =
 \left(n^{1/t}\right)^{1/R}
 =
 \left(\frac{t}{e\lambda}\right)^{1/R}.
\]
The condition
 $T\ge D\lambda n^{1/T}$
becomes
\[
 Rt\ge D\lambda\left(\frac{t}{e\lambda}\right)^{1/R},
\]
which is implied, after changing the absolute constant \(D\), by
\[
 R\left(\frac{t}{\lambda}\right)^{1-1/R}\ge D.
\]
The tail bound gives
\begin{eqnarray*}
 \Prb_h[\M(S,h)\ge Rt]
 &\le&
 C\left(
 \frac{\lambda n^{1/(Rt)}}{Rt}
 \right)^2\\
 &=&
 C\left(
 \frac{\lambda}{Rt}
 \left(\frac{t}{e\lambda}\right)^{1/R}
 \right)^2\\
 &\le&
 \frac{C}{R^2}
 \left(\frac{\lambda}{t}\right)^{2-2/R}.
\end{eqnarray*}
This proves the corollary.
\end{proof}

Finally, integrating the preceding tail gives a leading-constant comparison
with the fully independent scale in the sparse large-load regime.

\begin{corollary}[Expectation in the sparse large-load regime]
\label{cor:general-expectation}
Assume that
\[
 \frac{t}{\lambda}\to\infty,
\]
where \(t=t(m,n)\) is defined by
\[
 t\ln\left(\frac{t}{e\lambda}\right)=\ln n.
\]
Then
\[
 \E_h[\M(S,h)]
 \le
 (1+o(1))t.
\]
\end{corollary}

\begin{proof}
Let
\[
 \rho:=\frac{t}{\lambda}.
\]
By assumption, \(\rho\to\infty\). Choose a fixed constant \(F>0\) large enough,
and set
\[
 R_0:=1+\frac{F}{\ln\rho}.
\]
Then
\[
 1-\frac1{R_0}
 =
 \frac{R_0-1}{R_0}
 =
 \frac{F+o(1)}{\ln\rho}.
\]
Hence
\[
 R_0\rho^{1-1/R_0}
 =
 e^F(1+o(1)).
\]
Choosing \(F\) large enough ensures that
 $R_0\rho^{1-1/R_0}\ge D$
for all sufficiently large \(n\). Since the function
 $R\mapsto R\rho^{1-1/R}$
is increasing for \(R>0\), the condition
 $R\rho^{1-1/R}\ge D$
holds for every \(R\ge R_0\). Therefore
Corollary~\ref{cor:general-comparison} applies throughout the range
\(R\ge R_0\).

Using the tail-integral formula,
\begin{eqnarray*}
 \E_h[\M(S,h)]
 &=&
 \int_0^\infty \Prb[\M(S,h)\ge s]\,ds\\
 &\le&
 R_0t
 +
 t\int_{R_0}^\infty
 \Prb[\M(S,h)\ge Rt]\,dR.
\end{eqnarray*}
By Corollary~\ref{cor:general-comparison},
\[
 \E_h[\M(S,h)]
 \le
 R_0t
 +
 Ct\int_{R_0}^\infty
 \frac{1}{R^2}\rho^{-2+2/R}\,dR.
\]
Set
\[
 x:=1-\frac1R.
\]
Then \(dx=dR/R^2\), and the integral becomes
\[
 \int_{R_0}^\infty
 \frac{1}{R^2}\rho^{-2+2/R}\,dR
 =
 \int_{x_0}^{1}\rho^{-2x}\,dx,
 \qquad
 x_0:=1-\frac1{R_0}.
\]
Thus
\[
 \int_{x_0}^{1}\rho^{-2x}\,dx
 \le
 \frac{\rho^{-2x_0}}{2\ln\rho}.
\]
Since
\[
 x_0=\frac{F+o(1)}{\ln\rho},
\]
we have
 $\rho^{-2x_0}=e^{-2F+o(1)}$.
Therefore
\[
 \int_{x_0}^{1}\rho^{-2x}\,dx
 =
 O\left(\frac1{\ln\rho}\right)
 =
 o(1).
\]
Also
\[
 R_0=1+O\left(\frac1{\ln\rho}\right)=1+o(1).
\]
It follows that
\[
 \E_h[\M(S,h)]\le (1+o(1))t.
\]
\end{proof}

Thus, in the regime \(t/\lambda\to\infty\), binary linear hashing matches the
fully independent expected maximum-load scale up to a \(1+o(1)\) factor. This
improves the constant-factor tail and expectation bounds of the general
Jaber--Kumar--Zuckerman theorem in the sparse large-load range. It does not
address the dense balancing regime, where the maximum load is only a constant
factor above the average load \(\lambda\).

\end{document}